\newif \ifLong 
\Longtrue  

\documentclass[conference,letterpaper,10pt]{IEEEtran}

\newif \ifisit
\isitfalse

\usepackage[table]{xcolor}
\usepackage{times}
\usepackage{epsfig}
\usepackage[cmex10]{amsmath}
\usepackage{amsthm}
\usepackage{amsfonts}
\usepackage{graphicx}
\usepackage{amssymb}
\usepackage{amstext}
\usepackage{latexsym}
\usepackage{color,colortbl}
\usepackage{ifthen}
\usepackage{multirow}
\usepackage{verbatim}
\usepackage{array,tabularx}
\usepackage{arydshln}
\usepackage[mathscr]{euscript}
\usepackage{accents}
 \usepackage{cite}
\usepackage{hhline}
\usepackage{caption}
\usepackage{subcaption}
\usepackage{enumerate}
\usepackage{xcolor}
\usepackage{mathtools}
\usepackage{url}
\usepackage{xparse}
\usepackage{makecell}
\usepackage{varwidth}
\usepackage{bm}
\usepackage{arydshln}

\usepackage{hyperref}
\usepackage[shortlabels]{enumitem}
\usepackage{cleveref}
\hypersetup{
    colorlinks=true,
    linkcolor=blue,
    filecolor=magenta,      
    urlcolor=cyan,
    pdftitle={Overleaf Example},
    pdfpagemode=FullScreen,
}

\usepackage{algorithm}
\usepackage{algpseudocode}

\usepackage{comment}
\usepackage{wrapfig}

\usepackage{caption}
\usepackage{float}
\usepackage{booktabs}

\usepackage{dirtytalk}

\usepackage{mathtools}

\usepackage{booktabs}

\usepackage{float}

\usepackage{multirow}

\newcolumntype{C}[1]{>{\centering\let\newline\\\arraybackslash\hspace{0pt}}m{#1}}

\usepackage[normalem]{ulem}

\usepackage{amsmath,pgfplots,amssymb}
\usepackage{graphicx}

\newtheorem{theorem}{Theorem}

\newtheorem{proposition}{Proposition}

\theoremstyle{definition}

\newtheorem{definition}{Definition}

\theoremstyle{definition}
\newtheorem{remark}{Remark}

\theoremstyle{definition}

\newcommand{\interior}[1]{%
  {\kern0pt#1}^{\mathrm{o}}%
}

\newcommand{\cry}{\operatorname{Hash Crypt}} 
 
\newcommand{\veri}{\operatorname{Hash Verify}}

\usetikzlibrary{matrix,decorations.pathreplacing}

\newcommand{\rev}[1]{\textcolor{black}{#1}}

\definecolor{DarkGreen}{rgb}{0.1,0.5,0.1}
\definecolor{DarkRed}{rgb}{0.5,0.1,0.1}
\definecolor{DarkBlue}{rgb}{0.1,0.1,0.5}
\definecolor{DarkPurple}{rgb}{0.5,0.2,0.5}
\definecolor{DarkTurquoise}{rgb}{0.1,0.5,0.5}

\definecolor{beaublue}{rgb}{0.74, 0.83, 0.9}
\definecolor{coolblack}{rgb}{0.0, 0.18, 0.39}
\definecolor{apricot}{rgb}{0.98, 0.81, 0.69}
\definecolor{burntorange}{rgb}{0.8, 0.33, 0.0}
\definecolor{blue-violet}{rgb}{0.54, 0.17, 0.89}
\definecolor{byzantium}{rgb}{0.44, 0.16, 0.39}
\definecolor{brilliantrose}{rgb}{1.0, 0.33, 0.64}
\definecolor{cerisepink}{rgb}{0.93, 0.23, 0.51}
\definecolor{cobalt}{rgb}{0.0, 0.28, 0.67}
\definecolor{bostonuniversityred}{rgb}{0.8, 0.0, 0.0}

\definecolor{ao(english)}{rgb}{0.0, 0.5, 0.0}

\usepackage[utf8]{inputenc} 
\usepackage[T1]{fontenc}
\usepackage{url}
\usepackage{ifthen}
\usepackage{cite}

\newcommand{\off}[1]{}

\begin{document}

\title{Error Correction Capabilities of Non-Linear Cryptographic Hash Functions\vspace{-0.4cm}}



\author{Alejandro Cohen\IEEEauthorrefmark{1} and Rafael G. L. D'Oliveira\IEEEauthorrefmark{2}\\
\IEEEauthorrefmark{1}Faculty of Electrical and Computer Engineering, Technion, Israel, Email: alecohen@technion.ac.il\\
\IEEEauthorrefmark{2}SMSS, Clemson University, USA, Email: rdolive@clemson.edu
\vspace{-0.4cm}}


\maketitle


\begin{abstract}
Linear hashes are known to possess error-correcting capabilities. However, in most applications, non-linear hashes with pseudorandom outputs are utilized instead. It has also been established that classical non-systematic random codes, both linear and non-linear, are capacity achieving in the asymptotic regime. Thus, it is reasonable to expect that non-linear hashes might also exhibit good error-correcting capabilities.

In this paper, we show this to be the case. Our proof is based on techniques from multiple access channels. As a consequence, we show that Systematic Random Non-Linear Codes (S-RNLC) are capacity achieving in the asymptotic regime. We validate our results by comparing the performance of the Secure Hash Algorithm (SHA) with that of Systematic Random Linear Codes (S-RLC) and S-RNLC, demonstrating that SHA performs equally.

\end{abstract}



\section{Introduction}


Non-linear Cryptographic Hash Functions (NL-CHF), such as the Secure Hash Algorithms (SHA) \cite{burrows1995secure,eastlake2001us,penard2008secure}, have become a prevalent tool in a wide range of applications such as digital signatures, password protection, SSL handshakes, and data integrity checks \cite{preneel1993analysis,preneel1994cryptographic,rogaway2004cryptographic,sobti2012cryptographic}. One of the critical properties of these hash functions is collision resistance, which requires it to be computationally infeasible to find two distinct messages that produce the same hash value. This property plays a crucial role in ensuring the authenticity of data. For example, to validate a downloaded document, the user can calculate its message digest using the hash function and compare it with the originally stored message digest. If the two match, the document can be considered authentic.

However, as illustrated in Figure~\ref{fig:Sys}, many practical systems often involve the transmission of raw input messages \rev{ and the stored message digest} over noisy channels, which requires error correction. Similarly, in some cases, malicious actors can also partially manipulate downloaded messages \cite{sarwate2010coding,zhang2010p,dey2019sufficiently}.

The relationship between linear hashes and error-correction codes has been well established in the literature, as exemplified in works such as \cite{781407,harari1997hcc,ryabko2020linear,doi:10.1142/S1793830923500702}. To the best of our knowledge, this connection has only been studied for linear schemes. Thus, since the popular CHF, like SHA, are not linear,  the current prevalent systems in practice consist of two stages at the transmission side: first applying a non-linear hash function followed by a linear error correction coding scheme.

\begin{figure}[!t]
    \centering
    \includegraphics[trim={0cm 0cm 0cm 0cm},clip, width = 1\columnwidth]{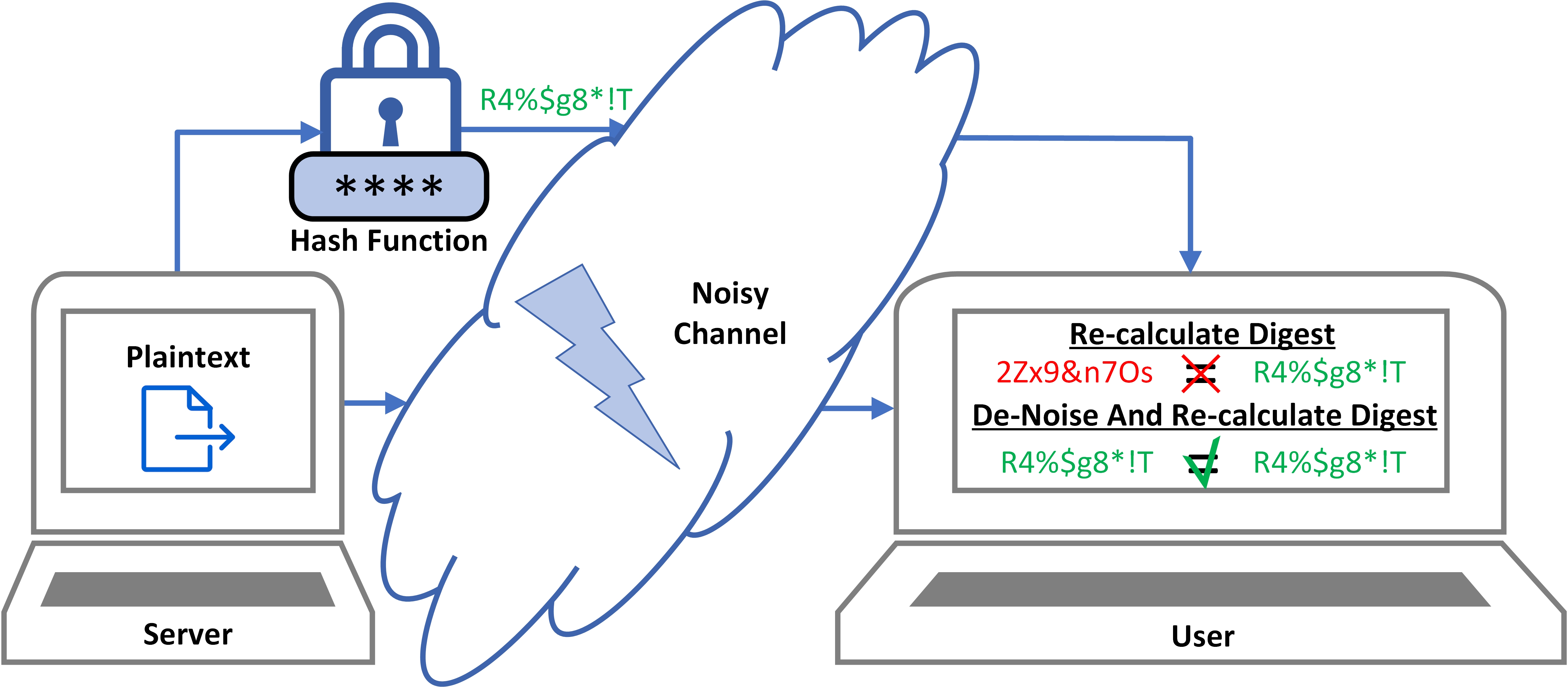}
    \caption{Reliable digital signature verification over a noisy channel, one source, Alice, one legitimate destination, Bob.}
    \label{fig:Sys}
    \vspace{-0.5cm}
\end{figure} 

In this paper, we show that commonly-used NL-CHF, such as SHA, possess forward error correction capabilities. Thus, when a NL-CHF is applied over a noisy communication channel, the authentication by the hash algorithm can be utilized not only for error detection, but also for error correction. This is not unexpected, since nonlinear random codes are capacity achieving in the asymptotic regime. In Theorem \ref{tho:deterministic}, we utilize techniques from multiple access channels \cite{cover2012elements,el2011network} to show that NL-CHF achieves capacity in the asymptotic regime. As a consequence, Systematic Random Non-Linear Codes (S-RNLC) also achieve capacity in the asymptotic regime.

Our scheme serves to expand the available options for system designers that utilize hash algorithms in their selection for error correction techniques. Instead of two stages, as previously discussed, for example, the transmitter can proceed with one single stage using only NL-CHF$^{\ref{footnote:keys}}\hspace{-0.05cm}$ schemes while achieving \emph{reliable authenticity validation} at the user at high data rates. \rev{The only change in the scheme proposed is on the decoding part (e.g., at the user in Figure~\ref{fig:Sys}). All the rest, including the transmitter and NL-CHF operations, are done identically as in the common existing solutions and schemes in the literature.} Based on the techniques presented in \cite{cohen2022partial}, we provide a new \rev{practical} joint error correction and hash check scheme for the user that combines error correction using an efficient Guessing Random Additive Noise Decoding (GRAND) \cite{duffy19GRAND}, with a hash function algorithm for authentication in an intermediate stage of the GRAND decoding algorithm, as illustrated in Figure~\ref{fig:scheme}.\footnote{In \cite{cohen2022aes}, an analogous decoder was utilized to show that cryptosystems with pseudorandom outputs like the Advanced Encryption Standard \cite{pub2001197,nechvatal2001report}, which offer secure data transmission, also possess error correcting capabilities.}

\begin{figure*}[!t]
    \centering
    \includegraphics[trim={0cm 0cm 0cm 0cm},clip, width = 1\textwidth]{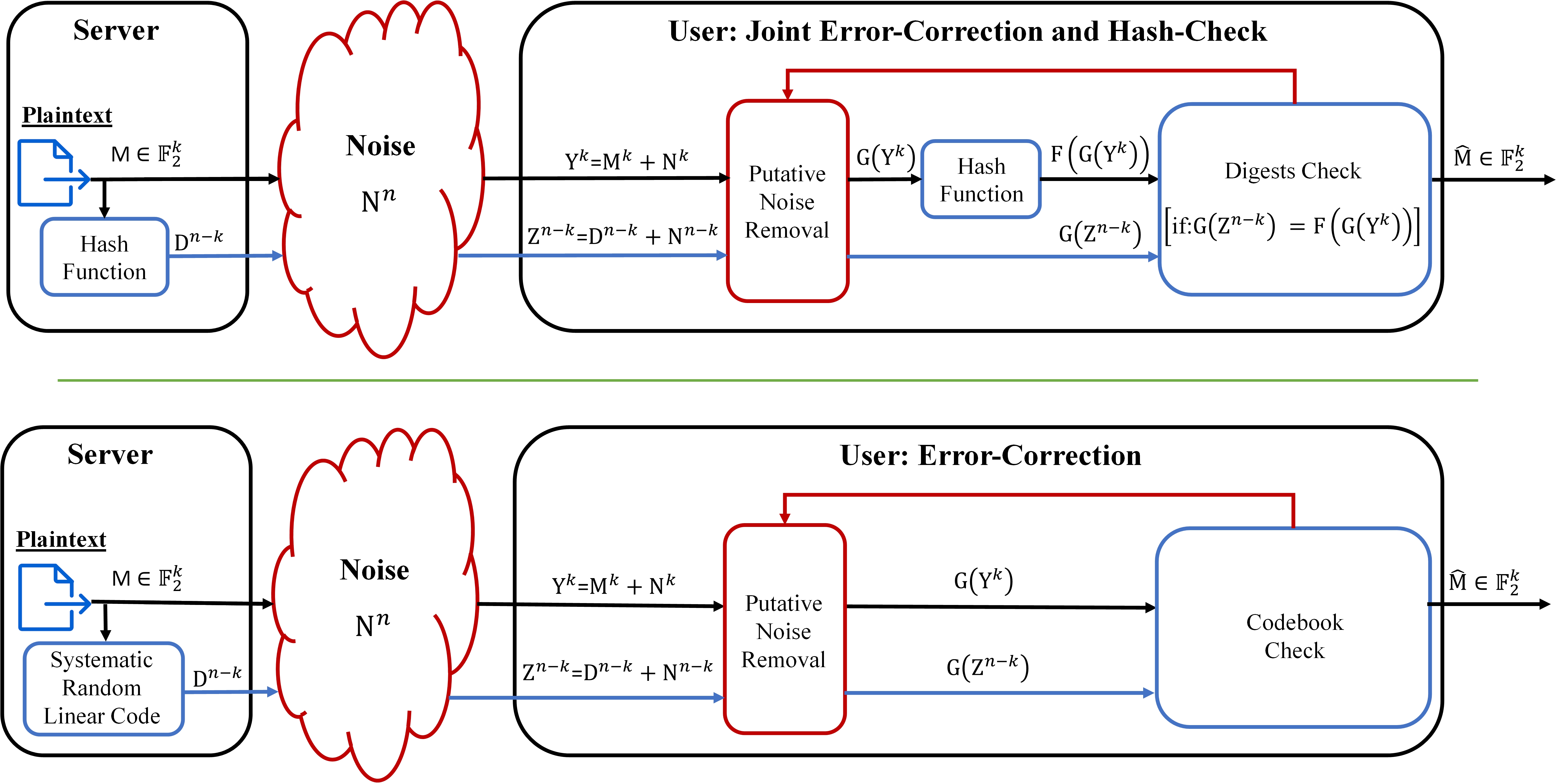}
    \caption{Proposed non-linear cryptographic hash functions with error correction capabilities, e.g., SHA (upper figure) and a Systematic RLC decoded with GRAND (bottom figure).}
    \label{fig:scheme}
   \vspace{-0.5cm}
\end{figure*}

Finally, we validate our results by comparing the performance of SHA1 \cite{national1993secure,pub2012secure}, utilizing the effective joint error correction and hash check scheme we propose for messages transmitted over the Additive white Gaussian noise (AWGN) channel, with that of Systematic Random Linear Codes (S-RLC), which are known to be capacity achieving and Systematic Random Non-Linear Codes (S-RNLC). Our results show that, in practical scenarios, the performance of SHA achieves the same error correction performance as S-RLC and S-RNLC.

The organization of this work is as follows. In Section~\ref{sec:system}, we describe the system model and the cryptographic hash functions and provide definitions used in this work. The proposed cryptographic hash function as an error-correcting coding scheme and the joint error correction and hash check decoder are presented in Section~\ref{sec:cry_ec}. Section~\ref{sec:main} demonstrates the error-correction performance of the proposed scheme using SHA1. In Section~\ref{sec:Perror}, we provide the probability of error analysis for the proposed NL-CHF as an error-correcting coding scheme. Finally, we conclude this work in Section~\ref{sec:conc}.

\section{Preliminaries}\label{sec:system}

\subsection{Setting}
Cryptographic hash functions, particularly secure hash algorithms, are widely used today for efficient digital signature verification, password hashing, SSL handshake, integrity checks, etc. To focus on the main contribution of the work, we consider the setting of the verification of digital signatures, as illustrated in Figures~\ref{fig:Sys} and~\ref{fig:scheme}. However, the approach proposed herein can be used for all the above problems in which NL-CHF is used.\footnote{\rev{We do note that digital signature schemes, using hashing functions, typically depend on keys available to legitimate parties in the hashing stage, to avoid sign messages at third party. To simplify the presentation and focus on the main contributions, in this paper, we assume that if keys are used in the cryptographic hash function, $F(\cdot)$, they are available at the legitimate parties. \label{footnote:keys}}}    

In the verification of digital signatures, we have a legitimate user, Bob, \rev{who} wishes to verify the authenticity of a downloaded document/file $M \in \mathbb{F}_{2}^{k}$ from a legitimate server, Alice. In this setting, Alice, in a preliminary stage, with an irreversible hash function $F(\cdot)$, produces and stores a fixed length message digest $D \in \mathbb{F}_{2}^{n-k}$ for each possible $M\in \mathbb{F}_{2}^{k}$. Then, when Bob requests a file from Alice, she sends both the original file requested and the appropriate message digest. \rev{We assume that Alice transmits the original file and the message digest to Bob over a noisy channel\footnote{\rev{To simplify the proof analysis,  we focus on the case where the message digest $D$ is also transmitted over a noisy channel. However, the proposed solution also works where $D$ is transmitted over a noiseless channel.}}. For the noisy communication at Bob, we consider an additive white Gaussian noise (AWGN) channel with noise~$N^n= (\bar{N}^{k},\tilde{N}^{n-k})$. The noise $N^n\sim\mathcal{N}\left(0,\sigma^2\right)$} is independent and identically distributed and is drawn from a normal distribution of zero mean with variance $\sigma^2$ \cite{cover2012elements}. \rev{We denote by $Y^k$ and $Z^{n-k}$ the demodulated noisy observation at Bob of the original file and the message digest, respectively, such that $Y^k=M^k+N^k$ and  $Z^{n-k}=D^{n-k}+N^{n-k}$.} 

\rev{The goal in this work is to design a decoding algorithm $\hat{M}(Y^k,Z^{n-k})$, such that observing the demodulated outcome of the noisy channel $Y^k$ and $Z^{n-k}$, Bob's decoder can identify the authentic legitimate document/file transmitted by Alice with high probability over $M$. That is, such that
\[
    \lim_{n \rightarrow \infty} P(\hat{M}(Y^k,Z^{n-k}) \neq M) \rightarrow 0.  
\]}

\subsection{Cryptographic Hash Functions}
A hash function is an algorithm that transforms an input value into a fixed-length output value, often referred to as a \say{hash} or a \say{digest}. This output value serves as a representation of the input. CHFs often require additional properties like preimage-resistance, second-preimage resistance, and collision-resistance. These properties can be defined in various ways (see \cite{rogaway2004cryptographic} for a complete discussion), but are not particularly relevant to our analysis. We are interested in CHFs that have pseudorandom outputs. Thus, we assume that the CHFs are a random oracle, a common assumption in the cryptographic literature \cite{bellare1993random,goldwasser1996lecture}.

\begin{definition} [Random oracle]
    A random oracle is an algorithm which “compresses” inputs value to a fixed-length output value as follows. For any input value, it checks if it has already generated a random fixed-length output. If so, it gives those random bits as output. If not, it creates random bits of the same fixed length, associates them with this input, and gives the random bits created as output.
\end{definition}


\begin{definition}[Cryptographic hash function]\label{Crypto_scheme}
A cryptographic hash function is made up of two algorithms:

\begin{itemize}
        \item A hashing algorithm $\cry(M)$, which acts as a random oracle\footnote{\rev{In practice, it is tested if one cannot distinguish the hash function from a real random oracle.}}, taking a plaintext message $M \in \mathbb{F}^{k}_{q}$ as input and generating a random digest signature $D \in \mathbb{F}^{n-k}_q$.
        \item A hash verification algorithm $\veri(M,D)$ that takes both the message $M$ and the digest $D$ as inputs, and outputs a Boolean value indicating whether the given digest $D$ is a valid signature for the plaintext message $M$.
    \end{itemize}
 That is, a cryptographic hash function requirement is that, 
\[
    \veri(M,\cry(M)) = \text{True}.
\]
\end{definition}

In the considered communication model, we deal exclusively with binary data and thus impose the restriction that the inputs and outputs of our hashing algorithm are binary vectors. Thus, the hash functions that we consider take the form $\cry : \mathbb{F}_2^k \rightarrow \mathbb{F}_2^{n-k}$.

\subsection{Random Functions and Random Codes} 
The connection between random functions and random codes, together with the noisy-channel coding theorem \cite{shannon1948mathematical} considered in this subsection, is the key to our main result given in Section~\ref{sec:cry_ec}.

\begin{definition}[Systematic Random Function] \label{def: random function}
A systematic random function is a function $f: \mathbb{F}_2^k \rightarrow \mathbb{F}_2^k \times \mathbb{F}_2^{n-k}$ chosen uniformly at random from the set of all functions from $\mathbb{F}_2^k$ to $\mathbb{F}_2^k \times \mathbb{F}_2^{n-k}$ that maintain the first $k$ bits unchanged.
\end{definition}
It is important to note that a pseudorandom function, as usually utilized in classical NL-CHF, is computationally indistinguishable from a systematic random function.

\begin{definition}[Random Code] \label{def: random code}
Let $k,n \in \mathbb{Z}_+$ be positive integers. A random code with rate $k/n$ is a subset $\mathcal{C} \subseteq \mathbb{F}_2^n$ chosen uniformly at random from the set of all subsets of $\mathbb{F}_2^n$ with cardinality less than or equal to $2^k$.
\end{definition}

In the asymptotic regime, i.e., when $n$ goes to infinity, random codes are capacity-achieving for the AWGN channels. In particular, this implies that as $n$ goes to infinity, the error probability of a code chosen uniformly at random with fixed rate $R=\frac{k}{n}<C$ goes to zero, \rev{where $C$ denotes} the capacity of the underlying AWGN channel. 

\begin{proposition}[Random Codes are Good] \label{prop: random codes good}
Let $C$ be the capacity of the underlying AWGN channel, and $\mathcal{C}$ be a random code with rate $R=\frac{k}{n}<C$. Then, as $n\rightarrow \infty$, the error probability $P_e(\mathcal{C})\rightarrow 0$, with high probability.
\end{proposition}

\begin{proof}
The proposition follows as a straightforward argument from the results in \cite[Chapter 7]{cover2012elements}.  The formal definitions for the error probabilities of a code we utilize are given in \cite[Chapter 7.5]{cover2012elements}.
\end{proof}

The key point we utilize for the proposed cryptographic hash functions as error correcting codes presented in the next section is that if $f:\mathbb{F}_2^k \rightarrow \mathbb{F}_2^k \times \mathbb{F}_2^{n-k}$ is a systematic random function, then $\mathcal{C} = f(\mathbb{F}_2^k) \subseteq \mathbb{F}_2^n$ is a systematic random code with rate $R = k/n$. Thus, as shown in Proposition \ref{prop: random codes good} for random codes, we show in Section~\ref{sec:Perror} that the systematic random code $\mathcal{C} = f(\mathbb{F}_2^k)$ has a high probability of being a good code \rev{as defined in Proposition~\ref{prop: random codes good}}.


\section{Cryptographic Hash Functions as Error Correcting Codes}\label{sec:cry_ec}
In this section, we show our joint hash check and error correction scheme utilizing traditional non-linear cryptographic hash functions with a uniform outcome. Recall that Alice needs to transmit over a noisy channel the requested file $M\in \mathbb{F}^{k}_{2}$ to Bob and a message digest $D\in \mathbb{F}^{n-k}_{2}$ for the signature verification. In Figure~\ref{fig:scheme}, we illustrate the encoding and decoding operations for the proposed non-linear hash function$^{\ref{footnote:keys}}$ with error correction capabilities and for the traditional error correction scheme using a systematic RLC code. 

We start by presenting the encoding process at Alice. In this setting, we assume that the messages are uniformly distributed with probability $p(m)$. For each possible message $m\in \mathbb{F}^k_2$, we assume Alice creates a uniformly distributed digest $d \in \mathbb{F}^{n-k}_2$ with probability $p(d)$ generated \off{(i.e., every message has its own digest) }by the cryptographic hash function $F(\cdot)$, given by      
\begin{equation*}
    \cry: F\left(M \in \mathbb{F}_2^k\right)  \rightarrow D \in \mathbb{F}_2^{n-k}.
\end{equation*}
Then, the message sequence $m_i$ and the digest sequence $d_j$ are transmitted over the noisy channel. That is, the proposed scheme is a non-linear systematic random code scheme where the multiplication gives the outcome of the first sequence of $k$ symbols with an identity matrix and the second sequence of $n-k$ symbols by the hash function.    

At Bob, the joint error-correction and hash \rev{check} scheme is given by
\begin{equation*}
\textstyle   \text{Joint-Dec-HashVerify}: Y \in \mathbb{F}_{2}^{k} \times Z \in \mathbb{F}_{2}^{n-k} \rightarrow \hat{M}\in\mathbb{F}_{2}^{k},
\end{equation*}
which maps the outcome noisy channel $Y^k$ and $Z^{n-k}$ to estimated authentic message transmitted $\hat{M}^{k}$.

\rev{To decode the messages, Bob utilizes a practical modified version of the GRAND decoder \cite{duffy19GRAND} as illustrated in Algorithm~\ref{alg:JD-D}.  Given the demodulated noisy channel outcome $Y^k$ and $Z^{n-k}$, Bob orders the putative noise effect sequences, $W^n= (\bar{W}^{k},\tilde{W}^{n-k}) \in \mathbb{F}_{2}^{n}$, from most likely to least likely (resp. to line~\ref{alg:JD-D_noise} in Algo.~\ref{alg:JD-D}). He then goes through the list, subtracting the putative noise effect $W^n$ from $(Y^k, Z^{n-k})$. Let $G(\cdot,W)$ denote the de-noising function, which subtracts the noise. Then, using the  verification cryptographic hash function, 
\[
\veri(G(Y^k,\bar{W}^k),G(Z^{n-k},\tilde{W}^{n-k})),
\]
Bob recalculates the hash digests from the de-noised sequence $Y^k$ obtained and compares it to the de-noised sequence $Z^{n-k}$ to establish data integrity (resp. to lines~\ref{alg:JD-D_dcry}-\ref{alg:JD-D_pading} in Algo.~\ref{alg:JD-D}). Hence, the file obtained is declared as the original data file requested if and only if
\begin{multline*}
   \hspace{-0.4cm}\text{for } \hat{D}^{n-k} \triangleq F(G(Y^k,\bar{W}^k))=F(G(M^k+\bar{N}^{k},\bar{W}^k)),\\
   \hspace{-0.0cm} \text{and }  \ddot{D}^{n-k} \triangleq G(Z^{n-k},\tilde{W}^{n-k}) = G(F(M^k)+\tilde{N}^{n-k},\tilde{W}^{n-k}),\\
   \off{\hspace{2.0cm} \text{and }   (M^k,D^{n-k},Y^{k},Z^{n-k})\in\mathcal{T}^{(n)}_{\epsilon},\\}
   \ddot{D}^{n-k} \text{ is equal to } \hat{D}^{n-k}. 
\end{multline*}
The first time this occurs, Bob declares $\hat{M} = G(Y^k,\bar{W}^k)$ as the transmitted message. As shown in \cite{duffy19GRAND}, this sequential de-noising proceed for the decoding procedure is a Maximum Likelihood (ML) decoder.}

\begin{algorithm}\small
\caption{\rev{Joint Error-Correction Decoding and Hash Check}}
\label{alg:JD-D}
 \vspace{-1.5mm}
\begin{flushleft}
Input:  hash function$^{\ref{footnote:keys}}$ $F(\cdot)$, noisy outcome channel $(Y^k,Z^{n-k})$ \newline
Output: $\hat{M}$, maximum likelihood decoding 
 \vspace{-1.5mm}
\end{flushleft}
\begin{algorithmic}[1]
\State $b \leftarrow 0$
\While{$b=0$}
\State $W^n \leftarrow$ next most likely noise effect\label{alg:JD-D_noise}\vspace{0.1cm}\newline
\vspace{0.1cm} $\overline{\underline{\hspace{0.5cm}\veri(G(Y^k,\bar{W}^k),G(Z^{n-k},\tilde{W}^{n-k}))^{\ref{footnote:keys}}\hspace{0.5cm}}}$
\State $\hat{D}^{n-k} \leftarrow F(G(Y^k,\bar{W}^k))$\label{alg:JD-D_dcry}
\State $\ddot{D}^{n-k} \leftarrow G(Z^{n-k},\tilde{W}^{n-k})$
\If{$\hat{D}^{n-k} == \ddot{D}^{n-k} $}
\State $\hat{M} \leftarrow G(Y^k,\bar{W}^k)$
\State $b \leftarrow 1$   
\State \Return $\hat{M}$
\EndIf\label{alg:JD-D_pading}
\EndWhile
\end{algorithmic}
\end{algorithm}

The main achievability result is given by the following theorem using the NL-CHF scheme and assuming that the hash function scheme satisfies Definition~\ref{def: random function} and~\ref{def: random code}.

\begin{theorem}\label{tho:deterministic}
Let $C$ be the capacity of the underlying AWGN channel, and suppose that $\cry$ is a random function. Then, if $\frac{k}{n}<C$, with high probability, the NL-CHF scheme can asymptotically transmit with an arbitrarily low probability of error at a rate of $\frac{k}{n}$.
\end{theorem}

\begin{proof}
In Section~\ref{sec:Perror}, the probability of error analysis for Theorem~\ref{tho:deterministic} is provided.
\end{proof}


\begin{remark}
It is important to note that the result presented in Theorem~\ref{tho:deterministic} and the accompanying analysis in Section~\ref{sec:Perror}, for the probability of error, hold implications for the field of forward error correction using non-linear systematic random codes. That is, those results show that S-RNLC schemes are capacity-achieving in the asymptotic regime.
\end{remark}

\section{SHA as an Non-Linear Error Correcting Code}\label{sec:main}
In this section, we demonstrate the error correction capabilities of non-linear cryptographic hash functions with uniform outcomes utilizing secure hash algorithms. The joint scheme is given in Section~\ref{sec:cry_ec} and is illustrated in Figure~\ref{fig:scheme}.\off{ The noise entropy we consider for the AWGN noisy channel is $h(N)$.} In particular, we show that in practical communication scenarios, traditional non-linear hash schemes, such as \rev{SHA1 and SHA-256}, have error correction capabilities that are similar to random codes. The simulation presented employs the AWGN channel model with \rev{Binary} Phase Shift Keying (BPSK) modulation. \rev{As a function of energy per information bit to noise power spectral density ratio, Eb/N0 \cite{enwiki:1121091972}, we use decoded Block Error Rate (BLER) as a performance metric.}

For the SHA1 cryptosystem, we use a standard FIPS scheme as given in \cite{national1993secure}. This is a common cryptographic hash function that belongs to the family of non-linear hash functions that can serve as error-correcting as defined in Section~\ref{sec:cry_ec}. For encoding codes, we use: (1)  Classical non-linear SHA1. (2) Systematic Random linear codes (S-RLC), which have known to be capacity achieving~\cite{gallager1973random}. (3) Systematic Random Non-Liniar Codes (S-RNLC). For the joint error correction and hash check decoding scheme, we use a \rev{practical} modified version of the GRAND decoder \cite{duffy19GRAND}, as proposed in Section~\ref{sec:cry_ec}. The \rev{software} implementations used for GRAND, SHA1\rev{, and SHA-256} are available at \cite{GRANDMATLAB} and \cite{SHA}, respectively.

\begin{figure}[!t]
\vspace{-0.2cm}
    \centering
    \includegraphics[trim=0.65cm 0cm 0cm 0cm,clip,width= 0.45\textwidth]{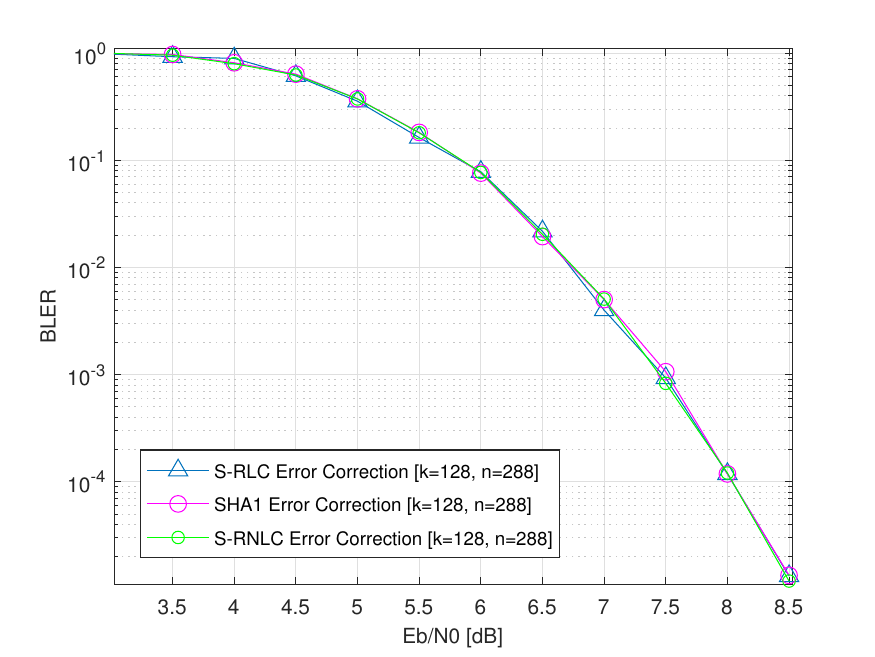}
    \caption{BLER vs. Eb/N0 for codes of $k=128$ and $n=288$ encoded with SHA1 as error correcting code and with Systematic RLC and Systematic RNLC. Here, both $M\in \mathbb{F}_{2}^{k}$ and $D\in \mathbb{F}_{2}^{n-k}$ are transmitted over the noisy channel. The joint error correction and hash check is performed with GRAND.}
    \label{fig:sim}
    \vspace{-0.5cm}
\end{figure}

\begin{figure}[!t]
\vspace{-0.2cm}
    \centering
    \includegraphics[trim=0.65cm 0cm 0cm 0cm,clip,width= 0.45\textwidth]{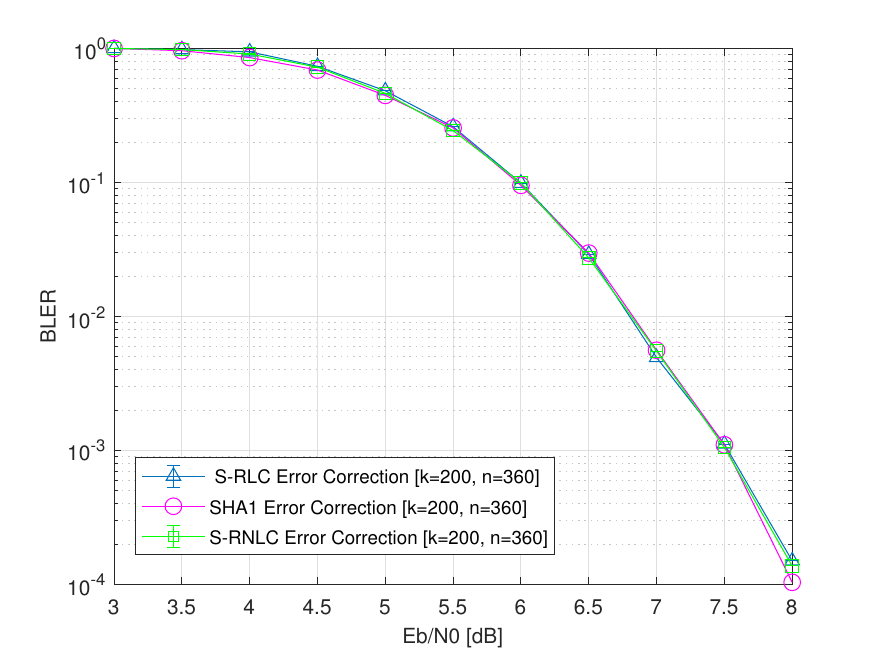}
    \caption{BLER vs. Eb/N0 for codes of $k=200$ and $n=350$ encoded with SHA1  as error correcting code and with Systematic RLC and Systematic RNLC. Here, both $M\in \mathbb{F}_{2}^{k}$ and $D\in \mathbb{F}_{2}^{n-k}$ are transmitted over the noisy channel. The joint error correction and hash check is performed with GRAND.}
    \label{fig:sim1}
    \vspace{-0.4cm}
\end{figure}

\off{\begin{figure}[!t]
    \centering
    \includegraphics[width= 0.49\textwidth]{SHA_256_and_RLC_k_128_n_160.eps}
    \caption{BLER vs. Eb/N0 for codes of length n = 160 and rate R = 0.8 encoded with SHA-256 as error correcting code and with Systematic RLC and Systematic RNLC. Here, as depicted in Figure~\ref{fig:scheme}, only $M\in \mathbb{F}_{2}^{k}$ is transmitted over the noisy channel. The joint decryption-decoding is performed with GRAND.}
    \label{fig:sim1}
    \vspace{-0.4cm}
\end{figure}
\begin{figure}[!t]
    \centering
    \includegraphics[width= 0.49\textwidth]{SHA1_and_RLC_k_128_n_148.eps}
    \caption{BLER vs. Eb/N0 for codes of length n = 148 and rate R = 0.86 encoded with SHA1 as error correcting code and with Systematic RLC. Here, as depicted in Figure~\ref{fig:scheme}, only $M\in \mathbb{F}_{2}^{k}$ is transmitted over the noisy channel. The joint decryption-decoding is performed with GRAND.}
    \label{fig:sim2}
    \vspace{-0.4cm}
\end{figure}}

In Figures~\ref{fig:sim}, we empirically contrast the error correction capabilities of the proposed approach using SHA1 with that of S-RLC and S-RNLC. We show the BLER vs. Eb/N0 for SHA1 and S-RLC codes with $k=128$ and length $n=288$. The blue lines show the performance of S-RLC codes, the purple lines of SHA1 hash algorithms as an error-correcting code, and the green lines show the performance of S-RNLC codes. The decoding is performed with GRAND. Compared to S-RLC  and S-RNLC codes, the schemes tested obtain almost the same performance in practical scenarios. In Figures~\ref{fig:sim1}, with rates $k=200$ and $n=360$, the same performance comparison is presented as in Figure~\ref{fig:sim}. We can notice that similar performances are obtained using all the tested coding schemes at different rates and lengths. Moreover, \rev{it is important to note that} forward error correction performance using SHA-256 was also tested, and the results obtained are, again, as with SHA1, similar to those using S-RNLC and S-RLC.

\section{Probability of Error Analysis (Theorem \ref{tho:deterministic})}\label{sec:Perror}
We now provide the probability of error analysis for Theorem~\ref{tho:deterministic} \rev{using typical decoder\off{\footnote{\rev{We do note that although typical decoding is suboptimal, compared to the proposed practical modified GRAND in Section~\ref{sec:cry_ec} (that is ML decoder), it is simple to analyze and still archives all the rates below capacity \cite{cover2012elements}.}}}}. \rev{Although typical decoding considered in the literature \cite{cover2012elements} and in this section is suboptimal, compared to the practical modified GRAND proposed in Section~\ref{sec:cry_ec} (that is ML decoder), it is simple to analyze and still archives all the rates below capacity with error probability $P_{e}^{n} \rightarrow 0$ as $n \rightarrow \infty$.} We assume that the NL-CHF scheme satisfies Definitions~\ref{def: random function} and~\ref{def: random code}. The error analysis follows the results and the techniques given in \off{\cite[Chapter 4.5]{el2011network} and} \cite[Chapter 15.3]{cover2012elements} for multiple access channels. We provide the adapted proof herein, adopting the scope, terminology, and notations to the NL-CHF as an error-correcting coding scheme proposed in this work.  

\rev{Let $\mathcal{T}^{(n)}_{\epsilon}$ denote the set of typical $(M^k,D^{n-k},Y^{k},Z^{n-k})$ sequences.} 

Without loss of generality, we assume that $(m_{i=1},d_{j=1})$ were requested by Bob and transmitted then by Alice. We have an error if either the correct $m_1$ and $d_1$ sequences are not typical with $Y^k$ and $Z^{n-k}$ sequences or for $i \neq 1$ and $j \neq 1$ incorrect sequences $m_i$ and $d_j$ there are typical with $Y^{k}$ and $Z^{n-k}$. Thus we define the error events by
\begin{align*}
& \mathcal{E}^{c}_{11} = \left \{ (m_1, d_1, y, z) \notin \mathcal{T}^{(n)}_{\epsilon} \right\},\\
& \mathcal{E}_{i1} = \left \{ (m_i, d_1, y, z) \in \mathcal{T}^{(n)}_{\epsilon}  \text{ for some } m_i \neq m_1\right\},\\
& \mathcal{E}_{1j} = \left \{ (m_1, d_j, y, z) \in \mathcal{T}^{(n)}_{\epsilon}  \text{ for some } d_j \neq d_1\right\},\\
& \mathcal{E}_{ij} = \left\{ (m_i, d_j, y, z) \in \mathcal{T}^{(n)}_{\epsilon} \text{ for some } m_i \neq m_1,d_j \neq d_1 \right\},
\end{align*}

Let $Q(\cdot)$ denote the conditional probability given that Alice transmitted the set of sequences $(m_1,d_1)$.
Thus, by the union bound of the error events, we have
\begin{multline}\label{eq:Perror}
    P_{e}^{n}\leq Q(\mathcal{E}^{c}_{11}) +  Q(\mathcal{E}_{i1}) + Q(\mathcal{E}_{1j})+ Q(\mathcal{E}_{ij}).
\end{multline}

We will now analyze the error events. From the asymptotic equipartition property \cite[Chapter 3]{cover2012elements}, we have $Q(\mathcal{E}^{c}_{11})\rightarrow 0$ for $n \rightarrow \infty$. For the sequences $(m_i,d_1)$ when $i \neq 1$ the joint pmf is given by $p(m)p(d)p(y,z|d)$, thus we have 
\begin{align*}
&\hspace{-0.3cm} Q(\mathcal{E}_{i1}) = Q\left((m_i, d_1, y, z) \in \mathcal{T}^{(n)}_{\epsilon}\right)\\
&\hspace{-0.3cm} = \sum_{(m, d) \in \mathcal{T}^{(n)}_{\epsilon}} p(m)p(d,y,z)\\
&\hspace{-0.3cm} \leq |\mathcal{T}^{(n)}_{\epsilon}|2^{-n(H(M)-\epsilon)}2^{-n(H(D,Y,Z)-\epsilon)}\\
&\hspace{-0.3cm} \leq \off{|\mathcal{T}^{(n)}_{\epsilon}|}2^{-n(H(M)+H(D,Y,Z)-H(M,D,Y,Z)-3\epsilon)}\\
&\hspace{-0.3cm} = 2^{-n(I(M;D,Y,Z)-3\epsilon)}\\
&\hspace{-0.3cm} = 2^{-n(I(M;Y,Z|D)-3\epsilon)},
\end{align*}
where the last two \rev{equalities} follow as we assume that the uniform random output $D^{n-k}$ of the hash function $F(\cdot)$ is statistically independent from message $M^k$. Such that we have, $I(M;D,Y,Z) = I(M;D) + I(M;Y,Z|D) = I(M;Y,Z|D)$.

For the sequences $(m_1,d_j)$ when $j \neq 1$ the joint pmf is given by $p(m)p(d)p(y,z|m)$, thus similarly, we have, 
\[
Q(\mathcal{E}_{1j}) \leq 2^{-n(I(D;Y,Z|M)-3\epsilon)},
\]
and for the sequences $(m_i,d_j)$ when $j \neq 1$ the joint pmf is given by $p(m)p(d)p(y,z)$ we have,  
\[
Q(\mathcal{E}_{ij}) \leq 2^{-n(I(M,D;Y,Z)-4\epsilon)}.
\] 
\rev{Finally}, for any arbitrary $\epsilon >0$, the conditional probability terms for the errors events above in \eqref{eq:Perror} are $Q(\cdot) \rightarrow 0$ as $n \rightarrow \infty$. Hence, for 
\[
 k < nI(M,D;Y,Z),
\]
since $P_{e}^{n} \rightarrow 0$ as $n \rightarrow \infty$, we can conclude that, on average, the error probability to identify the non-correct message by joint hash check and error correction scheme proposed is negligible. Hence, we achieve the bound on the rate provided in Theorem~\ref{tho:deterministic}. Namely, for \rev{$M \in \mathbb{F}_{2}^{k}$ and $D \in \mathbb{F}_{2}^{n-k}$}, the NL-CHF scheme can asymptotically transmit with an arbitrarily low probability of error at a rate of $\frac{k}{n}<C$.

\section{\rev{Conclusions}}\label{sec:conc}
In this work, we identify a family of non-linear
cryptographic hash functions with a uniform output that have error correction capabilities. In particular, we show that classical SHA algorithms, which are widely used, have error correction capabilities that allow for reliable transmission of requested plaintext over a noisy channel. The proposed approach opens up a new area of exploration for many applications in which hash functions are used. 
\bibliographystyle{IEEEtran}
\bibliography{refs.bib}

\end{document}